\documentclass[11pt, a4paper]{amsart}

\usepackage{xcolor}
\definecolor{MyBlue}{cmyk}{1,0.13,0,0.63}
\definecolor{MyGreen}{cmyk}{0.91,0,0.88,0.52}
\newcommand{\mylinkcolor}{MyBlue}
\newcommand{\mycitecolor}{MyGreen}
\newcommand{\myurlcolor}{black}

\usepackage{hyperref}
\hypersetup{%
  bookmarksnumbered=true,bookmarksopen=false,%
  plainpages=false,
  linktocpage=true,%
  colorlinks=true,breaklinks=true,%
  linkcolor=\mylinkcolor,citecolor=\mycitecolor,urlcolor=\myurlcolor,%
  pdfpagelayout=OneColumn,%
  pageanchor=true,%
}

\usepackage{graphicx}
\usepackage{amsmath, amsthm,  amsfonts, amscd}
\usepackage{amssymb}
\usepackage{url}
\usepackage{fullpage}
\usepackage{mathtools}
\usepackage[all]{xy}
\usepackage{slashed}
\usepackage{multirow}

\newtheorem{theorem}{Theorem}
\newtheorem*{thm*}{Theorem}
\newtheorem{corollary}{Corollary}
\newtheorem{lemma}{Lemma}
\newtheorem{proposition}[theorem]{Proposition}

\theoremstyle{definition}
\newtheorem{definition}{Definition}

\theoremstyle{remark}
\newtheorem{remark}{Remark}

\newcommand{\End}{\ensuremath{\mathrm{End}}}

\newcommand{\wh}{\ensuremath{\widehat}}

\newcommand{\R}{\ensuremath{\mathbb{R}}}
\newcommand{\N}{\ensuremath{\mathbb{N}}}
\newcommand{\Z}{\ensuremath{\mathbb{Z}}}

\newcommand{\C}{\ensuremath{\mathbb{C}}}
\newcommand{\T}{\ensuremath{\mathbb{T}}}

\def\calT{\mathcal{T}}

\def\calK{\mathcal{K}}
\def\calB{\mathcal{B}}
\def\calH{\mathcal{H}}

\def\calA{\mathcal{A}}
\def\calN{\mathcal{N}}

\def\calS{\mathcal{S}}

\def\calV{\mathcal{V}}

\def\calU{\mathcal{U}}

\def\bP{\mathbf{P}}

\newcommand{\ol}{\overline}

\theoremstyle{definition}

\DeclareMathOperator{\Dom}{Dom}

\DeclareMathOperator{\Index}{Index}

\DeclareMathOperator{\Ker}{Ker}

\DeclareMathOperator{\Tr}{Tr}
\DeclareMathOperator*{\res}{res}

\newcommand{\rst}[1]{\ensuremath{{\mathbin\upharpoonright}%
\raise-.5ex\hbox{$#1$}}}

\makeatletter

\newcommand{\Rmnum}[1]{\expandafter\@slowromancap\romannumeral #1@}
\makeatother

\author{C. Bourne}
\address{Department Mathematik, Friedrich-Alexander-Universit\"{a}t Erlangen-N\"{u}rnberg, 
Cauerstra{\ss}e 11, 91058 Erlangen, Germany \emph{and} 
Advanced Institute for Materials Research, Tohoku University,
2-1-1 Katahira, Aoba-ku, Sendai, 980-8577, Japan}
\email{christopher.j.bourne@gmail.com}

\author{H. Schulz-Baldes}
\address{Department Mathematik, Friedrich-Alexander-Universit\"{a}t Erlangen-N\"{u}rnberg, 
Cauerstra{\ss}e 11, 91058 Erlangen, Germany}
\email{schuba@mi.uni-erlangen.de}

\date{\today}

\begin{document}

\begin{abstract}
Recent work by Prodan and the second author showed that weak invariants 
of topological insulators can be described using Kasparov's 
$KK$-theory. In this note, a complementary 
description using semifinite index theory is given. This provides 
an alternative proof of the index formulae for 
weak complex topological phases using 
the semifinite local index formula. Real invariants and the 
bulk-boundary correspondence are also briefly considered.
\end{abstract}

\title{Application of semifinite index theory to 
weak topological phases}

\maketitle

\section{Introduction}
The application of techniques from the index theory of operator 
algebras to 
systems in condensed matter physics has given fruitful results, 
the quantum Hall effect being a key early example~\cite{Bellissard94}. 
More recently, $C^*$-algebras and their $K$-theory (and $K$-homology) 
have been applied to topological insulator systems, see for 
example~\cite{BCR15, GSB15,Kellendonk15,Kubota15b, PSBbook,Thiang14}.

\vspace{.1cm}

The framework of $C^*$-algebras is able to encode disordered systems 
with arbitrary (possibly irrational) magnetic field strength, something that standard methods
in solid state physics are unable to do. Furthermore, by considering 
the geometry of a dense subalgebra of the weak closure of the observable algebra, 
one can derive index formulae that relate
physical phenomena, such as the Hall conductivity, to an index of a Fredholm operator.

\vspace{.1cm}

Topological insulators are special materials which behave 
as an insulator in the interior (bulk) of the system, but have 
conducting modes at the edges of the system going along with non-trivial
topological invariants in the bulk \cite{SRFL}.
Influential work by Kitaev suggested that these properties are 
related to the $K$-theory of the momentum space of a free-fermionic system~\cite{Kitaev09}.

\vspace{.1cm}

Recent work by Prodan and the second author considered  
so-called `weak' topological phases of topological insulators~\cite{PSBKK}. 
In the picture without disorder or magnetic flux, a topological phase is classified 
by the real or complex $K$-theory of the torus $\T^d$ of 
dimension $d$. Relating Atiyah's $KR$-theory \cite{Ati} to the $K$-theory of $C^*$-algebras and 
then using the 
Pimsner--Voiculescu sequence with trivial action allows us to 
compute the relevant $K$-groups explicitly,
\begin{equation}\label{eq:periodic_weak_phase}
   KR^{-n}(\T^d,\zeta) \;\cong \; KO_n(C(i\T^d)) \;\cong \; KO_n(C^*(\Z^d)) \;\cong \; 
   \bigoplus_{j=0}^d \binom{d}{j} KO_{n-j}(\R)\;.
\end{equation}
Here $n$ labels the universality class as described in \cite{Kitaev09,BCR15} and 
$C(i\T^d)$ is the real $C^*$-algebra 
$\{f\in C(\T^d,\C)\,:\, \ol{f(x)} = f(-x)\}$, which naturally encodes 
the involution $\zeta$ on $\T^d$.
The `top degree' term $KO_{n-d}(\R)$ is said to represent the strong 
invariants of the topological insulator and all lower-order terms are called 
weak invariants.

\vspace{.1cm}

Bounded and complex Kasparov modules were used to provide a 
framework to compute weak invariants in the case of magnetic field 
and (weak) disorder in~\cite{PSBKK}. A geometric identity is 
used there to derive a local formula for the weak invariants.
The purpose of this paper is to 
provide an alternative proof of this result using 
semifinite spectral triples and, in particular, the semifinite 
local index formula in \cite{CPRS2,CPRS3}. This shows 
the flexibility of the operator algebraic approach and complements 
the work in~\cite{PSBKK}.

\vspace{.1cm}

The framework employed here largely follows from previous work, namely~\cite{BKR1}, where 
a Kasparov module and semifinite spectral triple were constructed 
for a unital $C^*$-algebra $B$ with a twisted $\Z^k$-action and 
invariant trace. Therefore the
main task here is the computation of the resolvent cocycle that 
represents the (semifinite) Chern character and its application 
to weak invariants. Furthermore, the bulk-boundary correspondence 
proved in~\cite{PSBbook,BKR1} also carries over, which allows us to 
relate topological pairings of the system without edge to pairings 
concentrated on the boundary of the sample.

\vspace{.1cm}

\subsubsection*{Acknowledgements}
We thank our collaborators, Alan Carey, Johannes Kellendonk, Emil Prodan and 
Adam Rennie, whose work this builds from. We also thank the anonymous referee, whose 
careful reading and suggestions have improved the manuscript. We are partially supported 
by the DFG grant SCHU-1358/6 and C. B. also thanks the Australian Mathematical 
Society and the Japan Society for the Promotion of Science for financial support.


\section{Review: Twisted crossed products and semifinite index theory}
\subsection{Preliminaries}
Let us briefly recall the basics of Kasparov theory that are
needed for this paper; a more comprehensive treatment can be 
found in~\cite{Blackadar, PSBKK}. Due to the anti-linear 
symmetries that exist in topological phases, 
both complex and real spaces and algebras  are considered.

\vspace{.1cm}

Given a real or complex right-$B$ $C^*$-module $E_B$, we will denote by 
$(\cdot\mid\cdot)_B$ the $B$-valued inner-product and by $\End_B(E)$ 
the adjointable endomorphisms on $E$ with respect to this inner product. 
The rank-$1$ operators $\Theta_{e,f}$, $e,f\in E_B$, are defined such that
$$
   \Theta_{e,f}(g)\; =\; e\cdot (f\mid g)_B\;, \qquad e,f,g\in E_B\;.
$$
Then $\End_B^{00}(E)$ denotes the span of such rank-$1$ operators.
The compact operators on the module, $\End_B^0(E)$, is the norm closure 
of $\End^{00}_B(E)$. We will often work with $\Z_2$-graded algebras and 
spaces and denote by $\hat\otimes$ the graded tensor product (see~\cite[Section 2]{Kasparov80}
and \cite{Blackadar}). Also see~\cite[Chapter 9]{Lance} for the basic theory 
of unbounded operators on $C^*$-modules.

\begin{definition}
{\sl Let $A$ and $B$ be $\Z_2$-graded real (resp. complex) $C^*$-algebras.
A real (complex) unbounded Kasparov module $(\calA, {}_\pi{E}_B, D)$ is a
$\Z_2$-graded real (complex) $C^*$-module ${E}_B$, a graded 
homomorphism $\pi:A \to \End_B({E})$, and an
unbounded self-adjoint, regular and odd operator $D$ such that for all $a\in \calA\subset A$, 
a dense $\ast$-subalgebra,
\begin{align*}
  & [D,\pi(a)]_\pm \,\in\, \End_B({E})\;,   &&\pi(a)(1+D^2)^{-1/2}\,\in\, \End_B^0({E})\;.
\end{align*}
For complex algebras and spaces, one can also remove the gradings, in 
which case the Kasparov module is called odd (otherwise even).
}
\end{definition}
We will often omit the representation $\pi$ when the left-action is unambiguous.
Unbounded Kasparov modules represent classes in the $KK$-group 
$KK(A,B)$ or $KKO(A,B)$~\cite{BJ83}. 

\vspace{.1cm}

Closely related to unbounded Kasparov modules are semifinite spectral 
triples.
Let $\tau$ be a fixed faithful, normal, semifinite trace on a von Neumann algebra 
$\calN$. Graded von Neumann algebras can be considered in an 
analogous way to graded $C^*$-algebras, though the only graded 
von Neumann algebras we will consider are of the form $\calN_0 \hat\otimes \End(\calV)$, 
with $\calN_0$ trivially graded and $\End(\calV)$ the graded operators on a 
finite dimensional and $\Z_2$-graded Hilbert space $\calV$. 
We denote by $\calK_\calN$ the $\tau$-compact
operators in $\calN$, that is, the norm closed ideal generated by the 
projections $P\in\calN$ with $\tau(P)<\infty$. For graded von Neumann algebras, 
non-trivial projections $P\in\calN$ are even, though the grading 
$\mathrm{Ad}_{\sigma_3}$ on $M_2(\calN)$ gives a grading on $M_n(\calK_\calN)$.
\begin{definition}
{\sl Let $\calN$ be a graded semifinite von Neumann algebra with trace $\tau$.
A semifinite spectral triple $(\calA,\calH,D)$ is given by a $\Z_2$-graded 
Hilbert space $\calH$, 
a graded $\ast$-algebra $\calA\subset\calN$ with $C^*$-closure $A$ and 
a graded representation on 
$\calH$, together with a densely defined odd 
unbounded self-adjoint operator $D$ affiliated to $\calN$ such that
\begin{enumerate}
  \item $[D,a]_\pm$ is well-defined on $\Dom(D)$ and extends to a bounded operator on $\calH$ for 
  all $a\in\calA$,
  \item $a(1+D^2)^{-1/2}\in\calK_\calN$ for all $a\in A$.
\end{enumerate}
}
\end{definition}

For $\calN=\calB(\calH)$ and $\tau = \Tr$, one recovers the 
usual definition of a spectral triple.

\vspace{.1cm}

If $(\calA,E_B,D)$ is an unbounded Kasparov module and the right-hand algebra 
$B$ has a faithful, semifinite and norm lower semicontinuous
trace $\tau_B$, then one can often construct a semifinite spectral 
triple using results from~\cite{LN04}. We follow this route in 
Section~\ref{sec-semifiniteconst} below. The converse is always true, 
namely a semifinite spectral triple gives rise to a class in $KK(A,C)$ 
with $C$ a subalgebra of $\calK_\calN$~\cite[Theorem 4.1]{KNR}. If $A$ is 
separable, this algebra $C$ can be chosen to be separable as well \cite[Theorem 5.3]{KNR}, 
but in a largely ad-hoc fashion.
Becasuse we first construct a Kasparov module
and subsequently build a semifinite spectral triple, one obtains 
more explicit control on the image of the semifinite index pairing defined next
(see Lemma~\ref{lemma:semifinite_is_KK} below). Therefore the algebra $C$
is not required here (as in \cite[Proposition 2.13]{CGRS2}) 
to assure that the range of the semifinite index pairing is 
countably generated, i.e. a discrete 
subset of $\R$.

\vspace{.1cm}

Complex semifinite spectral triples $(\calA,\calH,D)$ with $\calA$ 
trivially graded can be paired with $K$-theory classes in $K_\ast(\calA)$
via the semifinite Fredholm index. 
If $\calA$ is Fr\'{e}chet and stable under the holomorphic functional 
calculus, then $K_\ast(\calA) \cong K_\ast(A)$ and the pairings extend 
to the $C^*$-closure.
Recall that an operator 
$T\in \calN$ that is invertible modulo $\calK_\calN$ has semifinite Fredholm 
index
$$
   \Index_\tau(T) \;= \; \tau(P_{\Ker(T)})\; -\; \tau(P_{\Ker(T^*)})\;,
$$
with $P_{\Ker(T)}$ the projection onto $\Ker(T)\subset\calH$.
\begin{definition}
{\sl Let $(\calA,\calH,D)$ be a unital complex semifinite spectral triple relative to 
$(\calN,\tau)$ with $\calA$ trivially graded and $D$ invertible. 
Let $p$ be a projector in $M_n(\calA)$, which 
represents $[p]\in K_0(\calA)$ and $u$ a unitary in $M_n(\calA)$ 
representing $[u]\in K_1(\calA)$. In the even case, 
define $T_\pm = \frac{1}{2}(1\mp \gamma)T\frac{1}{2}(1\pm \gamma)$ 
with $\mathrm{Ad}_\gamma$ the grading on $\calH$. 
Then with $F = D|D|^{-1}$
and $\Pi =(1+F)/2$, the semifinite index pairing is 
represented by
\begin{align*}
&  \langle [p], (\calA,\calH,D)\rangle \;= \;\Index_{\tau\otimes \Tr_{\C^{n}}}\!(  {p}(F \otimes 1_{n})_+ p)\;, \qquad 
&\textrm{even case}\;, \\
&  \langle [u], (\calA,\calH,D) \rangle \;= \;\Index_{\tau\otimes \Tr_{\C^{n}}}\!\left( (\Pi \otimes 1_{n})u (\Pi \otimes 1_{n}) \right)\;, \qquad &\textrm{odd case}\;.
\end{align*}
}
\end{definition}

If $D$ is not invertible, we define the double spectral triple
$(\calA,\calH\oplus\calH,D_M)$ for $M>0$ and
relative to $(M_2(\calN),\tau\otimes\Tr_{\C^2})$, 
where the operator $D_M$ and the action of $\calA$ is given by
$$
D_M
\;=\; \begin{pmatrix} D & \;\;\;\;M \\ M & \;-D \end{pmatrix}\;,  \qquad
a\;\mapsto\;  \begin{pmatrix} a \;& 0 \\ 0\; & 0 \end{pmatrix}
\;,
$$
for all $a\in\calA$. If $(\calA,\calH,D)$ is graded by $\gamma$, then the double is graded by 
$\hat{\gamma} = \gamma \oplus (-\gamma)$. Doubling the spectral triple does not change the $K$-homology class 
and ensures that the unbounded operator $D_M$ is 
invertible~\cite{Connes85}.

\vspace{.1cm}

A unital semifinite spectral triple $(\calA,\calH,D)$ relative to 
$(\calN,\tau)$ is called $p$-summable if $(1+D^2)^{-s/2}$ is $\tau$-trace-class 
for all $s>p$, and smooth or $QC^\infty$ 
(for quantum $C^\infty$) if for all $a\in\calA$
$$
   a, [D,a] \,\in\, \bigcap_{n\geq 0} \Dom(\delta^n)\;,  \qquad \delta(T)\; =\; [(1+D^2)^{1/2},T]\;.
$$
If $(\calA,\calH,D)$ is complex, $p$-summable and $QC^\infty$, we can apply 
the semifinite local index formula~\cite{CPRS2,CPRS3} to compute the 
semifinite index pairing of $[x]\in K_\ast(A)$ with $(\calA,\calH,D)$ in terms of the resolvent cocycle. 
Because the resolvent cocycle 
is a local expression involving traces and derivations, it is usually easier to 
compute than the semifinite Fredholm index.

\subsection{Crossed products and Kasparov theory}
\label{sec-semifiniteconst}

\subsubsection{The algebra and representation}
Let us consider a $d$-dimensional lattice, so the Hilbert space $\calH = \ell^2(\Z^d)\otimes\C^n$, 
and a disordered family $\{H_\omega\}_{\omega\in\Omega}$ of Hamiltonians acting on $\calH$ indexed
by disorder configurations $\omega$ drawn from a compact space $\Omega$ equipped with 
a  $\Z^d$-action (possibly with twist $\phi$).
One can then construct the algebra of observables
$M_n(C(\Omega)\rtimes_\phi \Z^d)$. 
The family of Hamiltonians $\{H_\omega\}_{\omega\in\Omega}$ are associated 
to a self-adjoint element 
$H\in M_n(C(\Omega)\rtimes_\phi\Z^d)$, and we always assume that $H$ has a 
spectral gap at the Fermi energy. 
The Hilbert space fibres $\C^n$ and the 
matrices $M_n(\C)$ are often used to implement the symmetry operators that 
determine the symmetry-type of the Hamiltonian. However the matrices do 
not play an important role in the construction 
of the Kasparov modules and semifinite spectral triples we 
consider. 
Hence we will work with 
$C(\Omega)\rtimes_\phi\Z^d$, under the knowledge that this algebra can be tensored 
with the matrices (or compact operators) without issue.
The space $C(\Omega)$ can also encode 
a quasicrystal structure and depends on the example under consideration. 

\vspace{.1cm}

The twist $\phi$ is in general a twisting cocycle 
$\phi:\Z^d\times\Z^d\to \calU(C(\Omega))$ such that 
for all $x,\,y,\,z\in\Z^d$,
\begin{align*}
  \phi(x,y)\phi(x+y,z) \;=\;  \alpha_x(\phi(y,z))\phi(x,y+z)\;,  
  \qquad \phi(x,0)=\phi(0,x) \;=\; 1\;,
\end{align*}
see~\cite{PR89}. We also assume that $\phi(x,-x)=1$ for all $x\in\Z^d$ 
as in \cite{KR06} or \cite{PSBbook}, which still encompasses most
examples of physical interest.

\begin{remark}[Anti-linear symmetries, real algebras and twists]
Our model begins with a complex algebra acting on a complex Hilbert space. 
If the Hamiltonian satisfies anti-linear symmetries,
then we restrict to a real subalgebra of the complex algebra 
$C(\Omega)\rtimes_\phi\Z^d$ that is invariant under the induced real structure 
by complex conjugation. This procedure is direct for time-reversal 
symmetry, though modifications are needed for particle-hole symmetry  \cite{Thiang14,GSB15,Kellendonk15}. 
Such a restriction puts stringent constraints on the 
twisting cocycle $\phi$ and will often force the twist to be zero (e.g. if $\phi$ 
arises from an external magetic field). For this reason, in the real 
case, we will only consider untwisted crossed products 
$C(\Omega)\rtimes \Z^d$. We note that this may not encompass 
every example of interest, but we leave the the more general setting 
to another place.
\hfill $\diamond$
\end{remark}

Our focus is on weak topological invariants which have the interpretation of lower-dimensional 
invariants extracted from a higher-dimensional system. Using the 
assumption $\phi(x,-x)=1$, one can rewrite 
$C(\Omega)\rtimes_{\phi}\Z^{d} \cong \big(C(\Omega)\rtimes_{\phi}\Z^{d-k}\big)\rtimes_{\theta} \Z^k$ 
with a new twist $\theta$ \cite{PR89,KR06}.
Hence for $d$ large enough and $1\leq k\leq d$ one can study the lower-dimensional dynamics 
 and topological invariants of the $\Z^k$-action.

\vspace{.1cm}

With the setup in place, let $B$ be a unital separable $C^*$-algebra, real or complex, and consider 
the (twisted) crossed 
product $B\rtimes_\theta \Z^k$ with respect to 
a $\Z^k$-action $\alpha$. This algebra is generated by
the elements $b\in B$ and unitary operators $\{S_j\}_{j=1}^k$ such that 
$S^n=S_1^{n_1}\cdots S_k^{n_k}$ for $n=(n_1,\ldots,n_k)\in\Z^k$ satisfy 
$$
  S^n b \;= \;\alpha_n(b)S^n\;,  \qquad S^mS^n\;=\; \theta(n,m)S^{m+n}
$$
for multi-indices $n,m\in\Z^k$  and $\theta:\Z^k\times\Z^k\to \calU(B)$ the 
twisting cocycle. Let $\calA$ denote the algebra of  
$\sum_{n\in\Z^k} S^n b_n$, where $(\|b_n\|)_{n\in\Z^k}$ 
is in the discrete Schwartz-space $\calS(\ell^2(\Z^k))$. 
The full crossed product completion $B\rtimes_\theta \Z^k$ is denoted by $A$.
Following~\cite{BKR1,PSBKK} one can build an unbounded 
Kasparov module encoding this action. First let us take the standard $C^*$-module
$\ell^2(\Z^k)\otimes B=\ell^2(\Z^k,B)$ with right-action given by right-multiplication 
and $B$-valued inner product
$$
  ( \psi_1\otimes b_1 \mid \psi_2\otimes b_2)_B 
\;=\; \langle \psi_1,\psi_2\rangle_{\ell^2(\Z^k)} 
    \, b_1^*b_2
\;.
$$
The module $\ell^2(\Z^k,B)$ has the frame
$\{\delta_{m}\otimes 1_B\}_{m\in\Z^k}$ where $\{\delta_m\}_{m\in\Z^d}$ is the canonical 
basis on $\ell^2(\Z^k)$. Then an action on generators is defined by
\begin{align*}
   b_1\cdot( \delta_m \otimes b_2) &\;=\; \delta_m \otimes \alpha_{-m}(b_1)b_2\;, 
\\
  S^n\cdot(\delta_m \otimes b) &\;=\; \theta(n,m)\cdot \delta_{m+n}\otimes b  
    \;=\; \delta_{m+n}\otimes \alpha_{-m-n}(\theta(n,m))b\;.
\end{align*}
It is shown in~\cite{BKR1,PSBKK} that this left-action extends to an adjointable 
action of the crossed product on $\ell^2(\Z^k,B)$.

\subsubsection{The spin and oriented Dirac operators} \label{sec:spin_and_oriented_modules}

Using the position operators $X_j(\delta_m \otimes b) = m_j \delta_m\otimes b$ 
one can now build an unbounded Kasparov module. To put things together, the
real Clifford algebras $C\ell_{r,s}$ are used. They are generated by $r$ self-adjoint 
elements $\{\gamma^j\}_{j=1}^r$ with $(\gamma^j)^2 = 1$ and $s$ skew-adjoint 
elements $\{\rho^i\}_{i=1}^s$ with $(\rho^i)^2 = -1$. Taking the complexification 
we have $C\ell_{r,s}\otimes \C = \C\ell_{r+s}$.

\vspace{.1cm}

In the complex case and $k$ even, we may use the irreducible Clifford representation 
of $\C\ell_k= \text{span}_\C\{\Gamma^j\}_{j=1}^k$ on the (trivial) spinor 
bundle $\mathfrak{S}$ over $\T^k$ to construct the unbounded 
operator $\sum_{j=1}^k X_j \hat\otimes \Gamma^j$ on $\ell^2(\Z^k,B)\hat\otimes \mathfrak{S}$.
After Fourier transform, this is the standard Dirac operator on the spinor bundle over the torus. 
More concretely, $\mathfrak{S}\cong \C^{2^{k/2}}$ with $\{\Gamma^j\}_{j=1}^k$ 
self-adjoint matrices satsifying $\Gamma^i\Gamma^j + \Gamma^j \Gamma^i = 2\delta_{i,j}$.
For odd $k$, one proceeds similarly, but there are two irreducible representations of 
$\C\ell_k$ on $\mathfrak{S}\cong\C^{2^{(k-1)/2}}$.

\begin{proposition} \label{prop:spin_Dirac_module}
Consider a twisted $\Z^k$-action $\alpha,\theta$ on a complex $C^*$-algebra 
$B$. Let $A$ be the associated crossed product with dense subalgebra 
$\calA$ of $\sum_{n\in\Z^k}S^nb_n$  with  $(b_n)_{n\in\Z^k}$ Schwartz-class coefficients.
For $\nu = 2^{\lfloor \frac{k}{2} \rfloor}$, the triple
$$
 \lambda^{\!\mathfrak{S}}_k \;=\; \bigg( \calA, \, \ell^2(\Z^k,B)_B\hat\otimes \C^{\nu}, \, 
    \sum_{j=1}^k X_j \hat\otimes \Gamma^j \bigg)
$$
is an unbounded Kasparov module that is even if $k$ is even 
with grading given by $\mathrm{Ad}_{\Gamma_0}$ for $\Gamma_0 = (-i)^{k/2}\Gamma^1\cdots\Gamma^k$,
and specifying an element of $KK(A,B)$.
The triple $\lambda_k^{\!\mathfrak{S}}$ is odd (ungraded) if $k$ is odd, representing a class
in $KK^1(A, B)=KK(  A \hat\otimes \C\ell_1, B)$ which can be specified by a
graded Kasparov module
\begin{equation}
\label{eq-OddKKRep}
  \bigg( \calA \hat\otimes \C\ell_1, \, \ell^2(\Z^k,B) \otimes \C^{2^{(k-1)/2}}\hat\otimes \C^2, \, 
    \begin{pmatrix} 0 & -i\sum_{j=1}^k X_j\hat\otimes \Gamma^k \\ i\sum_{j=1}^k X_j\hat\otimes \Gamma^k & 0 \end{pmatrix} \bigg), 
\end{equation}
where the grading on $\C^2$ is given by conjugating with $\begin{pmatrix} 1 & 0 \\0 & -1 \end{pmatrix}$,
and $\sigma_1 = \begin{pmatrix} 0 & 1 \\ 1 & 0 \end{pmatrix}$ generates the 
left $\C\ell_1$-action.
\end{proposition}

\begin{proof}
The algebra $\calA$ is trivially graded and one computes that
$$
 \Big[ X_j, \sum_{m\in\Z^k}  S^m b_m \Big] \;=\; \sum_{m\in\Z^k} m_j S^m b_m
\;,
$$
which is adjointable for $(\|b_m\|)_{m\in\Z^k}$ in the Schwartz space over $\Z^k$. Therefore 
the commutator $[\sum_{j=1}^k X_j\hat\otimes\Gamma^j, a\hat\otimes 1_{\C^\nu}]$ 
is adjointable for $a\in\calA$. 
The operator $(1+|X|^2)^{-s/2}$ acts diagonally with respect to the frame 
$\{\delta_m\otimes 1_B\}_{m\in\Z^k}$ on $\ell^2(\Z^k, B)$. In particular,
$$
   (1+|X|^2)^{-1/2} \;=\; \sum_{m\in\Z^k} (1+|m|^2)^{-1/2} 
     \Theta_{\delta_m\otimes 1_B,\delta_m\otimes 1_B} \;,
$$
which is a norm convergent sum of finite-rank operators and so it is 
compact on $\ell^2(\Z^k,B)$. In particular, $(1+|X|^2)^{-1/2}\hat\otimes 1_{\C^{\nu}}$ 
is compact on $\ell^2(\Z^k,B)\hat\otimes \C^{\nu}$.
\end{proof}

The triple $\lambda_k^{\!\mathfrak{S}}$ is the unbounded representative of the bounded 
Kasparov module constructed in~\cite{PSBKK}. 
The (trivial) spin structure on the torus is used to construct the Kasparov module 
$\lambda_k^{\!\mathfrak{S}}$ from Proposition \ref{prop:spin_Dirac_module}. One 
can also use the torus' oriented structure.
Following~\cite[{\S}2]{Kasparov80}, we consider $\bigwedge^*\R^k$ (or complex), 
which is a graded Hilbert space such that 
$\End_\R(\bigwedge^*\R^k) \cong C\ell_{0,k}\hat\otimes C\ell_{k,0}$, 
where the action of $C\ell_{0,k}$ and $C\ell_{k,0}$ is generated by the 
operators
\begin{align*}
  &\rho^j(w) \;=\; e_j\wedge w - \iota(e_j)w\;, 
  &&\gamma^j(w) \;=\; e_j \wedge w + \iota(e_j)w\;,
\end{align*}
where $\{e_j\}_{j=1}^{k}$ denotes the standard basis of $\R^{k}$,
$w\in\bigwedge^*\R^{k}$ and $\iota(v)w$ the contraction of $w$ along $v$ 
(using the inner-product on $\R^k$). 
A careful check also shows that $\gamma^j$ and $\rho^k$ graded-commute.
The grading of $\bigwedge^*\R^k$ can be expressed in terms of the grading 
operator
$$
  \gamma_{\bigwedge^* \R^k} \;=\; (-1)^k \rho^1\cdots \rho^k \hat\otimes \gamma^k\cdots \gamma^1.
$$
Kasparov also constructs a diagonal action of $\mathrm{Spin}_{0,k}$ (and $\mathrm{Spin}_{k,0}$)  
on $\End_{\R}(\bigwedge^*\R^k)$ \cite[{\S}2.18]{Kasparov80}, though this will not be needed here.

\begin{proposition}[\cite{BKR1}, Proposition 3.2] \label{prop:crossed_prod_kas_mod}
Consider a $\Z^k$-action $\alpha$  on a 
real or complex $C^*$-algebra 
$B$, possibly twisted by $\theta$. Let $A$ be the associated crossed product with dense subalgebra 
$\calA$ of elements $\sum_{n} S^n b_n$ with Schwartz-class coefficients.
The data
\begin{equation} \label{eq:spectral_trip} 
 \lambda_k \;=\; \bigg( \calA \hat\otimes C\ell_{0,k},\, \ell^2(\Z^{k},B)_B \hat\otimes 
  \bigwedge\nolimits^{\!\ast}\R^{k},\, \sum_{j=1}^{k} X_j \hat\otimes \gamma^j \bigg)
\end{equation}
defines an unbounded $A\hat\otimes C\ell_{0,k}$-$B$ Kasparov module
and class in $KKO(A\hat\otimes C\ell_{0,k},B)$ which is also denoted $KKO^k(A,B)$. 
The $C\ell_{0,k}$-action is generated by the operators $\rho^j$.
In the complex case, one has to replace $\C\ell_k$ and $\bigwedge^*\C^k$ in the above formula.
\end{proposition}

For complex algebras and spaces, we have constructed two (complementary) Kasparov 
modules, $\lambda_k^{\!\mathfrak{S}}$ and $\lambda_k$. We have done this to 
 better align our results with existing 
literature on the topic, in particular~\cite{PSBbook, PSBKK}. 
In the case $k=1$, these Kasparov modules 
directly coincide. 

\vspace{.1cm}

For higher $k$, we can explicitly connect $\lambda_k^{\!\mathfrak{S}}$ 
and $\lambda_k$ by a Morita equivalence bimodule~\cite{Plymen86,LRV12}.
For $k$ even, there is an isomorphism $\C\ell_k \to \End(\C^{2^{k/2}})$ by Clifford 
multiplication. This observation implies that $\C^{2^{k/2}}$ is a $\Z_2$-graded Morita equivalence 
bimodule between $\C\ell_k$ and $\C$, where we equip $\C^{2^{k/2}}$ with 
a left $\C\ell_k$-valued inner-product ${}_{\C\ell_k}(\cdot\mid\cdot)$ such that 
${}_{\C\ell_k}(w_1\mid w_2)\cdot w_3 = w_1 \langle w_2, w_3 \rangle_{\C^\nu}$.
This bimodule gives an invertible 
class $[(\C\ell_k, \C^{2^{k/2}}_\C, 0)] \in KK(\C\ell_k,\C)$.  
One can take the external product of $\lambda_k^{\!\mathfrak{S}}$ with this class on the 
right to obtain (complex) $\lambda_k$. That is, 
$$
[\lambda_k^{\!\mathfrak{S}}] \hat\otimes_\C [(\C\ell_k, \C^{2^{k/2}}, 0)] 
\;=\; [\lambda_k] 
    \;\in\; KK(A\hat\otimes \C\ell_k,B)\;.
$$
Similarly $[\lambda_k^{\!\mathfrak{S}}] = [\lambda_k] \hat\otimes [(\C,(\C^{2^{k/2}})^*_{\C\ell_k},0)]$ 
with $(\C^{2^{k/2}})^*_{\C\ell_k}$ the conjugate module providing the 
inverse to $[(\C\ell_k, \C^{2^{k/2}}_\C, 0)]$, see \cite{RW} for more 
details on Morita equivalence bimodules.

\vspace{.1cm}

For $k$ odd we use the graded Kasparov module \eqref{eq-OddKKRep} instead of $\lambda_k^{\!\mathfrak{S}}$.
We can again compose this graded Kasparov module with the $KK$-class 
from the Morita equivalence bimodule
$( \C\ell_{k-1}, \C^{2^{(k-1)/2}}_\C, 0 )$. The external product 
gives $[\lambda_k]\in KK^k(A,B)$.
 Hence from an index-theoretic perspective, the Kasparov modules $\lambda_k^{\!\mathfrak{S}}$ 
and $\lambda_k$ are equivalent up to a normalisation coming from the spinor dimension.

\vspace{.1cm}

In the case of real spaces and algebras, a similar (but more involved) 
equivalence also holds for real spinor representations. Namely, 
for $\mathbb{K} =\R$, $\C$ or $\mathbb{H}$, there is a 
unique irreducible representation
$C\ell_{r,s} \to \End_{\mathbb{K}}(\mathfrak{S}_\mathbb{K})$ 
if $s-r+1$ is not a multiple of 
$4$, otherwise there are $2$ irreducible 
representations~\cite[Chapter 1, Theorem 5.7]{SpinGeo}. 
To relate these modules to $\bigwedge^*\R^k$, one also uses that 
$\C \cong \R^2$ and $\mathbb{H} \cong \R^4$. Obviously there are 
more cases to check in the real setting, but because we do not 
use the spin Kasparov module in the real case, the full details 
are beyond the scope of this paper.

\vspace{.1cm}

In order to consider weak invariants in the real case, we will often
go beyond the limits of semifinite index theory and will need to work with the
Kasparov modules and $KK$-classes directly. In such a setting, we prefer 
to work with the `oriented' Kasparov module $\lambda_k$ for 
several reasons:
\begin{enumerate}
  \item The oriented structure, $\bigwedge^*\R^k$, and its corresponding Clifford 
  representations is at the heart of Kasparov theory and, for example, plays a key role in the 
  proof of Bott periodicity~\cite[{\S}5]{Kasparov80} and 
  Poincar\'{e} duality~\cite[{\S}4]{KasparovNovikov}. 
  This is also evidenced in Theorem \ref{thm:bulk-edge} below (also compare with 
  \cite{FR15}, where to achieve factorisation of equivariant (spin) spectral triples, 
  a `middle module' is required that plays of the role of the complex Morita equivalence 
  linking $\lambda_k^{\!\mathfrak{S}}$ and $\lambda_k$ for complex algebras).
  \item The Clifford actions of $C\ell_{0,k}$ and $C\ell_{k,0}$ on $\bigwedge^*\R^k$ 
  are explicit. This makes the Clifford representations more amenable to 
  the Kasparov product as well as the Clifford index used to define 
  real weak invariants (see Section \ref{sec:real_pairings}).
\end{enumerate}

\subsubsection{Kasparov module to semifinite spectral triple}

Returning to the example $B=C(\Omega)\rtimes_\phi\Z^{d-k}$, it will be assumed that 
$\Omega$ possesses a probability measure $\bP$ that is invariant under the $\Z^d$-action 
and $\mathrm{supp}(\bP) = \Omega$. Hence $\bP$ induces a faithful trace on $C(\Omega)$ 
and $C(\Omega)\rtimes_\phi \Z^{d-k}$ by the formula
$$
    \tau\Big(\sum_{m\in\Z^{d-k}} {S}^m g_m \Big) 
\;= \;
\int_{\Omega} g_0(\omega)\,\mathrm{d}\bP(\omega)
\;.
$$
Thus, we will assume from now on that our generic algebra $B$ has a faithful and 
norm lower semicontinuous trace, $\tau_B$, that is 
invariant under the $\Z^k$-action.
This trace now allows to construct a semifinite spectral triple from the above Kasparov 
module. We first construct the GNS space $L^2(B,\tau_B)$ and consider the new
Hilbert space $\ell^2(\Z^k)\otimes L^2(B,\tau_B)$. Let us note that
$\ell^2(\Z^k)\otimes L^2(B,\tau_B) \cong \ell^2(\Z^k,B)\otimes_B L^2(B,\tau_B)$ 
so the adjointable action of $A=B\rtimes_\theta\Z^k$ on $\ell^2(\Z^k,B)$ extends 
to a representation of $A$ on $\ell^2(\Z^k)\otimes L^2(B,\tau_B)$.

\begin{proposition}[\cite{LN04}, Theorem 1.1]
Given $T\in\End_{B}(\ell^2(\Z^k,B))$ with $T\geq 0$, define
$$  
\Tr_\tau(T) \;=\; \sup_{I} \sum_{\xi\in I}\tau_B\!\left[( \xi\mid T\xi)_{B}\right]\;, 
$$
where the supremum is taken over all finite subsets $I\subset \ell^2(\Z^k,B)$ with
$\sum_{\xi\in I}\Theta_{\xi,\xi}\leq 1$. 
\begin{enumerate}
 \item Then $\Tr_\tau$ is a semifinite norm lower semicontinuous 
trace on the compact endomorphisms $\End_{B}^0(\ell^2(\Z^k,B))$
with the property $\Tr_{\tau}(\Theta_{\xi_1,\xi_2}) = \tau_B[(\xi_2\mid \xi_1)_{B}]$.
 \item Let $\calN$ be the von Neumann algebra $\End_{B}^{00}(\ell^2(\Z^k,B))''\subset 
\calB[\ell^2(\Z^k)\otimes L^2(B,\tau_B)]$. 
Then the trace $\Tr_\tau$ extends to a faithful semifinite trace on the positive cone $\calN_+$. 
\end{enumerate}
\end{proposition}
Recall that the operator $(1+|X|^2)$ acts diagonally on the frame $\{\delta_m \otimes 1_B\}_{m\in \Z^k}$,
so 
$$
  (1+|X|^2)^{-s/2} \;=\; \sum_{m\in\Z^k} (1+|m|^2)^{-s/2} \Theta_{ \delta_m \otimes 1_B, \delta_m \otimes 1_B}. 
$$
Using the properties $\Tr_\tau$, one can compute that
\begin{align*}
  \Tr_\tau \big((1+|X|^2)^{-s/2}\big) 
  &\;=\; \sum_{m\in\Z^k}(1+|m|^2)^{-s/2}\, 
        \tau_B((\delta_m\otimes 1_B\mid \delta_m\otimes 1_B)_B) \\
     &\;= \; \sum_{m\in\Z^k}(1+|m|^2)^{-s/2}\,\tau_B(1_B).
\end{align*}
This observation and a little more work gives the following result.

\begin{proposition}[\cite{BKR1}, Proposition 5.8] \label{prop:bulk_Kasmod_spec_trip}
For $\calA\subset B \rtimes_\theta \Z^k$ the algebra of operators 
$\sum_{n\in\Z^k}S^n b_n$ with Schwartz-class coefficients, the tuple
$$
\bigg(\calA\hat\otimes C\ell_{0,k},\, \ell^2(\Z^k)\otimes L^2(B, \tau_B)\hat\otimes \bigwedge\nolimits^{\!*}\R^k,\, 
   \sum_{j=1}^k X_j\otimes 1\hat\otimes \gamma^j  \bigg) 
$$
is a $QC^\infty$ and $k$-summable semifinite spectral triple relative to 
$\calN\hat\otimes \End(\bigwedge^*\R^k)$ with trace 
$\Tr_\tau \hat\otimes \Tr_{\bigwedge^*\R^k}$.
\end{proposition}

We have the analogous result for the spin Dirac operator.
\begin{proposition} \label{prop:spin_semifinite_trip}
The tuple
$$
  \bigg( \calA, \, \ell^2(\Z^k) \otimes L^2(B,\tau_B) \hat\otimes \C^\nu, \, 
    \sum_{j=1}^k X_j\otimes 1 \hat\otimes \Gamma^j \bigg)
$$
is a $QC^\infty$ and $k$-summable complex semifinite spectral triple relative to 
$\calN\hat\otimes \End(\C^\nu)$ with trace $\Tr_\tau \hat\otimes \Tr_{\C^\nu}$. The 
spectral triple is even if $k$ is even with grading operator 
$\Gamma_0 = (-i)^{k/2}\Gamma^1\cdots\Gamma^k$. The spectral triple is odd if $k$ is odd.
\end{proposition}

Therefore all hypotheses required to apply the semifinite local index formula are satisfied. 
Furthermore, the algebra $\calA$ is Fr\'{e}chet and stable under the holomorphic functional calculus. Therefore 
all pairings of $K_k(\calA)$ extend to pairings with $K_k(B\rtimes_\theta \Z^k)$.


\section{Complex pairings and the local index formula}
Let us now restrict to a complex algebra $A=B\rtimes_\theta\Z^k$, where $B$ is 
separable, unital and possesses a faithful, semifinite and 
norm lower semicontinuous trace $\tau_B$ 
that is invariant under the $\Z^k$-action. 
First, the semifinite index pairing is related to the `base algebra' $B$ 
and the dynamics of the $\Z^k$-action.
\begin{lemma} \label{lemma:semifinite_is_KK}
The semifinite index pairing of a class $[x]\in K_k(B\rtimes_\theta \Z^k)$ 
with the spin semifinite spectral triple from Proposition \ref{prop:spin_semifinite_trip}
can be computed by the $K$-theoretic composition
\begin{equation}\label{eq:refined_semifinite_pairing}
 K_k(B\rtimes_\theta \Z^k) \times KK^k(B\rtimes_\theta\Z^k,B) 
\;\to\; K_{0}(B) 
 \;  \xrightarrow{(\tau_B)_*} \;\R\; ,
\end{equation}
with the class in $KK^k(B\rtimes_\theta\Z^k,B)$ represented by 
$\lambda^{\!\mathfrak{S}}_k$ 
 from Proposition \ref{prop:spin_Dirac_module}. 
\end{lemma}
\begin{proof}
We start with the even pairing, with $p\in M_q(B\rtimes_\theta\Z^k)$ 
representing $[p]\in K_0(B\rtimes_\theta\Z^k)$.
Taking the double $X=X_M$ if necessary,
the semifinite index pairing is given by the semifinite index
$$
  \langle [p], [(\calA,\calH,X)] \rangle 
\,  =\, (\Tr_\tau\otimes\Tr_{\C^l})( P_{\Ker(p (X\otimes 1_q)_+ p)} ) - 
     (\Tr_\tau\otimes\Tr_{\C^l})( P_{\Ker(p (X\otimes 1_q)_+^* p)} )
\,,
$$
with $P_{\Ker(T)}$ the projection onto the kernel of $T$, 
$\Tr_{\C^l}$ the
finite trace from the spin structure
and the operator $X_+$ comes from the decomposition 
$X = \begin{pmatrix} 0 & \;X_- \\ X_+ & \;0 \end{pmatrix}$ due to the grading in even dimension.
Next we compute the Kasparov product in Equation \eqref{eq:refined_semifinite_pairing} 
following, for example, \cite[Section 4.3.1]{PSBKK}.
The product $[p]\hat\otimes_A [\lambda_k]\in KK(\C,B)$ is represented by the class of the 
Kasparov module
$$
   \left( \C,\, p\big(\ell^2(\Z^k,B)^{\oplus q}\big)\otimes \C^{2l},\, 
      \begin{pmatrix} 0 & p(X\otimes 1_q)_- p \\ p(X\otimes 1_q)_+ p & 0 \end{pmatrix} \right)\,, 
   \quad \gamma = \mathrm{Ad} \begin{pmatrix} 1 & 0 \\ 0 & -1 \end{pmatrix}\,.
$$
After regularising if necessary, 
$\Ker( p(X\otimes 1_q)_+ p)$ is a finitely generated and projective submodule of 
$p\big(\ell^2(\Z^k,B)^{\oplus q}\big)\otimes\C^l$ and the projection onto this submodule is 
compact (and therefore finite-rank). We can associate a $K$-theory class to this Kasparov module
by noting that 
$\End_B^0\big(p(\ell^2(\Z^k,B))^{\oplus q} \otimes\C^l\big) \cong B\otimes \calK$ 
and taking the difference 
$$
[P_{\Ker( p(X\otimes 1_q)_+ p)}] \;-\; [ P_{\Ker(p(X\otimes 1_q)_+^* p)}] \,\in\, K_0(B)
$$ 
Because $\Ker( p(X\otimes 1_q)_+ p)$ is finitely generated, the module 
$p(\ell^2(\Z^k,B)^{\oplus q}) \otimes\C^l$ has a
\emph{finite} frame $\{e_j\}_{j=1}^n$ such that 
$\sum_{j=1}^n \Theta_{e_j,e_j} = \mathrm{Id}_{\Ker( p(X\otimes 1_q)_+ p)}$. 
Taking the induced trace $(\tau_B)_\ast: K_0(B) \to \R$, one can use the properties 
of the dual trace $\Tr_\tau$ to note that
$$
   \tau_B\big(P_{\Ker( p(X\otimes 1_q)_+ p)}\big) 
\;=\; \sum_{j=1}^n \tau_B( (e_j\mid e_j)_B ) 
   \; =\;  \sum_{j=1}^n \Tr_\tau(  \Theta_{e_j,e_j} )\;.
$$
The right hand side is now a trace defined over 
$\End_B^{00}\big(p(\ell^2(\Z^k,B)^{\oplus q})\otimes\C^l)\subset \calN\hat\otimes \End(\C^l)$ and by construction it is 
the same as $(\Tr_\tau\otimes\Tr_{\C^l})( P_{\Ker(p (X\otimes 1_q)_+ p)})$.
An analogous result holds for $\Ker( p(X\otimes 1_q)_+^* p)$, so $(\tau_B)_\ast([p]\hat\otimes_{B\rtimes_\theta\Z^k} [\lambda_k] )$ 
is represented by
$$
 (\Tr_\tau\otimes\Tr_{\C^l})( P_{\Ker(p (X\otimes 1_q)_+ p)} ) \;- \;
    (\Tr_\tau\otimes\Tr_{\C^l})( P_{\Ker(p (X\otimes 1_q)_+^* p)} )
\;,
$$
and thus the pairings coincide.

\vspace{.1cm}

For the odd pairing, the same argument applies for $\Index_{\Tr_\tau}(\Pi u\Pi)$ with 
$\Pi$ the positive spectral projection of $X$ and $[u] \in K_1(B\rtimes_\theta\Z^k)$. For this, one 
has to appeal to the appendix of \cite{KNR} or \cite[Section 4.3.2]{PSBKK}. 
\end{proof}

\vspace{.1cm}

Lemma \ref{lemma:semifinite_is_KK} means that the semifinite pairing considered here has 
a concrete $K$-theoretic interpretation. 
In particular, we know that 
$\langle [x],[(\calA,\calH,X)] \rangle \subset \tau_B(K_0(B))$, which is 
countably generated for separable $B$.
This is one of the reasons we build a Kasparov module first and {then} 
construct a semifinite spectral triple via the dual trace $\Tr_\tau$. 

\begin{remark}
We may also pair $K$-theory classes with the Kasparov module $\lambda_k$ from 
Proposition \ref{prop:crossed_prod_kas_mod} by the composition
\begin{equation} \label{eq:oriented_complex_pairing}
  K_k(B\rtimes_\theta \Z^k) \times KK^k( B\rtimes_\theta \Z^k, B) 
   \;\to\; KK( \C\ell_{2k}, B) \;\xrightarrow{\;\cong\;}\; K_0(B) 
   \xrightarrow{(\tau_B)_\ast} \R
\end{equation}
where $KK( \C\ell_{2k}, B) \,{\cong}\, K_0(B)$ by stability 
and~\cite[{\S}6, Theorem 3]{Kasparov80}.
We can think of Equation \eqref{eq:oriented_complex_pairing} as the 
definition of the complex semifinite index pairing of $K$-theory with the
semifinite spectral triple from Proposition \ref{prop:bulk_Kasmod_spec_trip} 
over the graded algebra $B\rtimes_\theta\Z^k \hat\otimes \C\ell_k$. 
Indeed, in more general 
circumstances, the $K$-theoretic composition is how the semifinite pairing 
is defined, where in general 
one pairs with the class in $KK^k(A,C)$ with $C$ a subalgebra of 
$\calK_\calN$~\cite[Section 2.3]{CGRS2}.

\vspace{.1cm}

Equation \eqref{eq:oriented_complex_pairing} also has a natural 
analogue in the real case, namely
$$
 KO_k(B\rtimes \Z^k) \times KKO^k(B\rtimes \Z^k,B) \;\to\; 
  KKO(C\ell_{k,0}\hat\otimes C\ell_{0,k}, B) \;\xrightarrow{\cong}\; KO_0(B) 
   \xrightarrow{(\tau_B)_\ast} \R
$$
as $C\ell_{k,0}\hat\otimes C\ell_{0,k} \cong M_l(\R)$ which
is Morita equivalent to $\R$. 
Of course, we also want to pair our Kasparov module with elements in 
$KO_j(B\rtimes\Z^k)$ for $j\neq k$, and in this situation we use the 
general Kasparov product (see Section \ref{sec:real_pairings}).
\hfill $\diamond$
\end{remark}

To compute the local index formula, we first note some preliminary results.
\begin{lemma} \label{lem:residue_trace_formula_part1}
The function
$$
  \zeta(s) \;=\; \Tr_\tau\!\left(S^n b(1+|X|^2)^{-s/2}\right)\;, \qquad s>k\;,
$$
has a meromorphic extension to the complex plane with
$$  
\res_{s=k} \;\Tr_\tau\!\left(S^n b(1+|X|^2)^{-s/2}\right) 
\;=\; \delta_{n,0}\, {\mathrm{Vol}_{k-1}(S^{k-1})} 
   \tau_B(b)\;.
$$
\end{lemma}
\begin{proof}
We use the frame $\{\delta_m \otimes 1_B\}_{m\in\Z^k}$ for $\ell^2(\Z^k,B)$ and note that
$S^n b\cdot( \delta_m \otimes 1_B) = \delta_{m+n}\otimes \alpha_{-m-n}(\theta(n,m))\alpha_{-m}(b)$.
Computing, for $s>k$,
\begin{align*}
  \Tr_\tau\!&\left(S^n b(1+|X|^2)^{-s/2}\right)  \;=\; \Tr_\tau \bigg( S^n b \sum_{m\in\Z^k} (1+|m|^2)^{-s/2}
    \Theta_{\delta_m\otimes 1, \delta_m \otimes 1} \bigg) \\
    &\;=\; \sum_{m\in\Z^k} (1+|m|^2)^{-s/2} \Tr_\tau\!\left( \Theta_{\delta_{m+n}\otimes \alpha_{-m-n}(\theta(n,m))\alpha_{-m}(b), \delta_m\otimes 1} \right) \\
    &\;=\; \sum_{m\in\Z^k} (1+|m|^2)^{-s/2} \tau_B\big( \langle \delta_m,\delta_{n+m}\rangle_{\ell^2(\Z^k)} 
        \alpha_{-m-n}(\theta(n,m))\alpha_{-m}(b) \big) \\
    &\;= \;\delta_{n,0} \sum_{m\in\Z^k} (1+|m|^2)^{-s/2} \tau_B\big( \theta(0,m)b \big) \\
    &\;= \;\delta_{n,0}\, \tau_B(b) \sum_{m\in\Z^k} (1+|m|^2)^{-s/2} \\
    &\;=\; \delta_{n,0}\, \tau_B(b)\, \mathrm{Vol}_{k-1}(S^{k-1}) \, \frac{\Gamma\!\left(\frac{k}{2}\right) \Gamma\!\left(\frac{s-k}{2}\right)}{2\Gamma\!\left(\frac{k}{2}\right)}\;,
\end{align*}
where  the invariance of the $\alpha$-action in the trace was used.
By the functional equation for the 
$\Gamma$-function, $\zeta(s)$ has a meromorphic extension to the complex plane 
and is holomorphic for $\Re(s)>k$. Computing the residue obtains the result.
 \end{proof}

\vspace{.1cm}

Next let us note that any trace on $B$ can be extended to $\calA$ by defining
$$
  \calT\!\left(\sum_n S^n b_n\right) \;=\; \tau_B(b_0)\;,
$$
where $\calT$ is faithful and norm lower semicontinuous if $\tau_B$ is faithful and 
norm lower semicontinuous. A direct extension of Lemma \ref{lem:residue_trace_formula_part1} 
then gives that
\begin{equation}\label{eq:residue_trace_formula}
  \res_{s=k}\; \Tr_\tau\!\left(a(1+|X|^2)^{-s/2}\right)\; =\;  {\mathrm{Vol}_{k-1}(S^{k-1})} 
   \calT(a)\;, \qquad a\in\calA\;.
\end{equation}

\subsection{Odd formula}

We will compute the semifinite pairing with the spectral triple constructed 
from Proposition \ref{prop:spin_semifinite_trip}, which 
aligns our results with~\cite{PSBKK}. The equivalence between 
spin and oriented semifinite spectral triples means that we also obtain 
formulas for the pairing with the semifinite spectral triple from 
Proposition \ref{prop:bulk_Kasmod_spec_trip}, where the 
result would be the same up to a normalisation.

\vspace{.1cm}

Except for certain cases where specific results on the spinor trace 
of the gamma matrices are needed, we will write the trace $\Tr_\tau \hat\otimes \Tr_{\C^\nu}$ 
on the von Neumann algebra $\calN \hat\otimes \End(\C^\nu)$ as just $\Tr_\tau$.

\begin{theorem}[Odd index formula] \label{thm:higher_dim_chern_number_odd}
Let $u$ be a complex unitary in $M_q(\calA)$ and $X_\mathrm{odd}$ the complex semifinite
spectral triple from Proposition \ref{prop:spin_semifinite_trip} with $k$ odd. 
Then the semifinite index pairing is given by the formula
$$  \langle [u], [X_\mathrm{odd}] \rangle\; =\; C_k \sum_{\sigma\in S_k}(-1)^\sigma\, 
 (\Tr_{\C^q}\otimes\calT) \bigg(\prod_{i=1}^k u^* \partial_{\sigma(i)}u \bigg)\;, $$
where $C_{2n+1} = \frac{ -2(2\pi)^n n!}{i^{n+1}(2n+1)!}$, 
$\Tr_{\C^q}$ is the matrix trace on $\C^q$, $S_k$ is the permutation group on $\{1,\ldots,k\}$  
and $\partial_ja = -i[X_j,a]$ for any $a\in\calA$ and $j\in\{1,\ldots,k\}$.
\end{theorem}

Let us focus on the case $q=1$ and then extend to matrices by taking 
$(D\otimes 1_q)$ with $D=\sum_{j=1}^k X_j\otimes\Gamma^j$.
Because the semifinite spectral triple of Proposition 
 \ref{prop:spin_semifinite_trip} is smooth and with spectral dimension $k$, 
 the odd local index formula from \cite{CPRS2} gives 
$$ \langle [u],[X_\mathrm{odd}]\rangle 
\;= \;\frac{-1}{\sqrt{2\pi i}}\; \res_{r=(1-k)/2} 
\;\sum_{m=1,\text{odd}}^{2N-1} \! \phi_m^r(\mathrm{Ch}^m(u))\;, $$
where $u$ is a unitary in $\calA$, $N=\lfloor k/2\rfloor + 1$ and 
$$  \mathrm{Ch}^{2n+1}(u) 
\;= \;
(-1)^n n!\,u^*\otimes u \otimes u^*\otimes \cdots \otimes u\;, \hspace{0.5cm}(2n+2\text{ entries})
\;. 
$$
The functional $\phi_m^r$ is the resolvent cocycle from~\cite{CPRS2}. 
To compute the index pairing we recall the following important observation.

\begin{lemma}[\cite{BCPRSW}, Section 11.1] \label{prop:only_need_top_term_in_local_index_odd_case}
The only term in the sum $\sum\limits_{m=1,\mathrm{odd}}^{2N-1}\! \phi_m^r(\mathrm{Ch}^m(u))$ that 
contributes to the index pairing is the term with $m=k$.
\end{lemma}
\begin{proof}
We first note that the spinor trace on the Clifford generators is given by
\begin{equation} \label{eq:trace_of_grading}
 \Tr_{\C^\nu} (i^k \Gamma^1\cdots \Gamma^k) 
\;=\; 
(-i)^{\lfloor (k+1)/2 \rfloor} 2^{\lfloor (k-1)/2\rfloor}
\;,
\end{equation} 
and will vanish on any product of $j$ Clifford generators with $0<j<k$. 
The resolvent cocycle involves the spinor trace of terms
$$  
a_0R_s(\lambda)[D,a_1]R_s(\lambda)\cdots [D,a_m]R_s(\lambda)\,, \qquad R_s(\lambda) 
\;=\; 
(\lambda-(1+s^2+D^2))^{-1}
\;, 
$$
for $a_0,\ldots,a_m\in\calA$. Noting that $[D,a_l] = i\sum_{j=1}^k \partial_j a_l\otimes \Gamma^j$ 
and $R_s(\lambda)$ is diagonal 
in the spinor representation, it follows that the product  $a_0R_s(\lambda)[D,a_1]\cdots [D,a_m]R_s(\lambda)$ 
will be in the span of $m$ Clifford generators acting on 
$\ell^2(\Z^k)\otimes L^2(B,\tau_B)\hat\otimes\C^\nu$.  
Furthermore, the trace estimates ensure that each spinor component of $\phi_m^r$
$$  
\int_\ell \lambda^{-k/2-r} a_0(\lambda-(1+s^2+|X|^2))^{-1}\partial_{j_1}a_1 \cdots \partial_{j_m}a_m(\lambda-(1+s^2+|X|^2))^{-1}\,\mathrm{d}\lambda 
$$
is trace-class for $a_0,\ldots,a_m\in\calA$ and real part $\Re(r)$ sufficiently large. 
Hence for $0<m<k$, the spinor trace will vanish for $\Re(r)$ large and 
$\phi_m^r(\mathrm{Ch}^m(u))$ 
analytically extends as a function holomorphic in a 
neighbourhood of $r=(1-k)/2$ for $0<m<k$. Thus $\phi_m^r(\mathrm{Ch}^m(u))$ does 
not contribute to the index pairing for $0<m<k$.
 \end{proof}

\begin{proof}[Proof of Theorem \ref{thm:higher_dim_chern_number_odd}]
Lemma \ref{prop:only_need_top_term_in_local_index_odd_case} simplifies the semifinite index substantially, namely it is given by the expression
\begin{align*}
  \langle [u],[X_\mathrm{odd}]\rangle &\;=\;  \frac{-1}{\sqrt{2\pi i}} \;\res_{r=(1-k)/2} \;\phi_k^r(\mathrm{Ch}^k(u))
\;.
\end{align*}
Therefore one needs to compute the residue at $r=(k-1)/2$ of
$$  
 \mathcal{C}_k \int_0^\infty\! s^k\,\Tr_\tau \!\left( \int_\ell \lambda^{-k/2-r} u^* R_s(\lambda) [D,u] R_s(\lambda)[D,u^*]\cdots [D,u]R_s(\lambda) \,\mathrm{d}\lambda\right)\!\mathrm{d}s\;,  
$$
where $k=2n+1$ and the constant
$$
\mathcal{C}_k \;=\;   -\,\frac{(-1)^{n+1}n!}{(2\pi i)^{3/2}} \;\frac{\sqrt{2i}\; 2^{d+1} \Gamma(d/2+1)}{\Gamma(d+1)}\;
$$ 
comes from the definition of the resolvent cocycle, 
see~\cite[Section 3.2]{CGRS2}, and $\mathrm{Ch}^k(u)$.
To compute this residue we move all terms $R_s(\lambda)$ to the right, 
which can be done up to a function holomorphic at $r=(1-k)/2$. 
This allows us to take the Cauchy integral. We then observe that 
$\underbrace{[D,u][D,u^*]\cdots[D,u]}_{k\text{ terms}} \in \calA \otimes 1_{\C^\nu}$, 
so Lemma \ref{lem:residue_trace_formula_part1} implies that 
the zeta function
$$ \Tr_\tau\!\left(u^*{[D,u][D,u^*]\cdots[D,u]}(1+D^2)^{-z/2}\right) $$
has at worst a simple pole at $\Re(z)=k$. Therefore we can explicitly compute
\begin{align*}
  &\frac{-1}{\sqrt{2\pi i}} \;\res_{r=(1-k)/2} \;\phi_k^r(\mathrm{Ch}^k(u))  \\
   &\hspace{2cm}=\;  (-1)^{n+1} \,n! \,\frac{1}{k!}\,\tilde{\sigma}_{n, 0}\,\res_{z=k}
     \Tr_\tau\!\left(u^*[D,u][D,u^*]\cdots[D,u](1+D^2)^{-z/2}\right)
   \;, 
\end{align*}
where the numbers $\tilde{\sigma}_{n, j}$ are defined by the formula
$$  \prod_{j=0}^{n-1} (z+j+1/2) \;=\; \sum_{j=0}^n z^j \tilde{\sigma}_{n, j}\;. $$
Hence the number $\tilde{\sigma}_{n,0}$ is the coefficient of $1$ in the 
product $\prod_{l=0}^{n-1}(z+l+1/2)$. This is the product of all the non-$z$ 
terms, which can be written as
$$  (1/2)(3/2)\cdots (n-1/2) \;= \;\frac{1}{\sqrt{\pi}}\;\Gamma(k/2)\;. $$
Putting this back together, our index pairing can be written as
\begin{align*}
   \langle [u],[X_\mathrm{odd}]\rangle  &\;= \; (-1)^{n+1}\, \frac{n! \Gamma(k/2)}{k!\sqrt{\pi}}\; \res_{z=k}\; 
   \Tr_\tau\!\left(u^*[D,u][D,u^*]\cdots[D,u](1+D^2)^{-z/2}\right).
\end{align*}
We make use of the identity $[D,u^*] = -u^*[D,u]u^*$, which allows us to rewrite
\begin{align*}
   u^*\underbrace{[D,u][D,u^*]\cdots[D,u]}_{k=2n+1\text{ terms}} &
   \;=\; (-1)^n u^*[D,u]u^*[D,u]u^*\cdots u^*[D,u] \\
     &\;=\; (-1)^n \left( u^*[D,u]\right)^k\;.
\end{align*}
Recall that $[D,u] = \sum_{j=1}^k [X_j,u]\hat\otimes\Gamma^j = i\sum_{j=1}^k \partial_j(u)\hat\otimes\Gamma^j$, 
so applying this relation we have that
$u^*[D,u] = i\sum_{j=1}^k u^* \partial_j(u)\hat\otimes\Gamma^j$. Taking the $k$-th 
power
$$  
\left(u^*[D,u]\right)^k \;=\; i^k \sum_{J=(j_1,\ldots,j_k)} u^*(\partial_{j_1}u)\cdots  u^*(\partial_{j_k}u) \hat\otimes \Gamma^{j_1}\cdots \Gamma^{j_k} 
$$
where the sum is extended over all multi-indices $J$. Note that every term in 
the sum is a multiple of the identity of $\C^\nu$ and so has a non-zero spinor trace. 
Writing this product in terms of permutations,
$$  
  (-1)^n\left(u^*[D,u]\right)^k \;=\; (-1)^n i^k \sum_{\sigma\in S_k}(-1)^\sigma 
   \prod_{j=1}^k u^*(\partial_{\sigma(j)}u)\hat\otimes \Gamma^{j} \;,
$$
with $S_k$ is the permutation group of $k$ letters.
Let's put all this back together. 
\begin{align*}
  \langle [u], [X_\mathrm{odd}] \rangle &\;=\; (-1)^{n+1} \frac{n! \Gamma(k/2)}{k!\sqrt{\pi}}\res_{z=k} \Tr_\tau\!\left(u^*[D,u][D,u^*]\cdots[D,u](1+D^2)^{-z/2}\right) \\
    &\hspace{-1.1cm}= \;-\, \frac{n! \Gamma(k/2)}{k!\sqrt{\pi}}\;\res_{z=k}\; \Tr_\tau \bigg[i^k \bigg(\sum_{\sigma\in S_k}(-1)^\sigma \prod_{j=1}^k u^*(\partial_{\sigma(j)}u)\hat\otimes \Gamma^{j}\bigg) \!(1+D^2)^{-z/2} \bigg] \\
    &\hspace{-1.1cm}= \;-\,\frac{n! \Gamma(k/2) 2^{\lfloor (k-1)/2\rfloor}}{i^{\lfloor (k+1)/2 \rfloor}\;k!\,\sqrt{\pi} }\;\res_{z=k} \;\Tr_{\tau} \bigg(\sum_{\sigma\in S_k}(-1)^\sigma \prod_{j=1}^k \!u^*(\partial_{\sigma(j)}u) (1+|X|^2)^{-z/2} \bigg)\;,
\end{align*}
where we have used Equation \eqref{eq:trace_of_grading} and that 
$(1+D^2) = (1+|X|^2)\otimes 1_{\C^\nu}$. We can apply 
Equation \eqref{eq:residue_trace_formula} to reduce the formula to
\begin{align*}
   \langle [u], [X_\mathrm{odd}] \rangle  &
   \;=\; -\frac{n! \Gamma(k/2)\mathrm{Vol}_{k-1}(S^{k-1}) 2^{\lfloor (k-1)/2\rfloor}}
   {i^{\lfloor (k+1)/2 \rfloor} k!\sqrt{\pi}} \sum_{\sigma\in S_k}(-1)^\sigma \, 
   \calT \bigg( \prod_{i=1}^k u^* (\partial_{\sigma(i)}u) \bigg)\;.
\end{align*}
Now the identity $\mathrm{Vol}_{k-1}(S^{k-1}) = \frac{k\pi^{k/2}}{\Gamma(k/2+1)}$ allows to simplify
$$  
\frac{n! \Gamma(k/2)\mathrm{Vol}_{k-1}(S^{k-1}) 2^{\lfloor (k-1)/2\rfloor}}{i^{\lfloor (k+1)/2 \rfloor} k!\sqrt{\pi}} 
\;=\; \frac{ 2(2\pi)^n n!}{i^{n+1}(2n+1)!}\;,  
$$
for $k=2n+1$, and therefore
$$  
\langle [u], [X_\mathrm{odd}] \rangle \;= \;C_k \sum_{\sigma\in S_k}(-1)^\sigma\,  
  \calT \bigg( \prod_{i=1}^k u^* (\partial_{\sigma(i)}u) \bigg)\;, 
  \qquad C_{2n+1}\; =\; \frac{ -2(2\pi)^n n!}{i^{n+1}(2n+1)!}\;,
$$
which concludes the argument.
\end{proof}

\subsection{Even formula}

\begin{theorem}[Even index formula] \label{thm:higher_dim_chern_number_even}
Let $p$ be a complex projection in $M_q(\calA)$ and $X_\mathrm{even}$ the complex semifinite
spectral triple from Proposition \ref{prop:spin_semifinite_trip} with $k$ even. 
Then the semifinite index pairing can be expressed by the formula
\begin{equation*} 
  \langle [p], [X_\mathrm{even}]\rangle \;=\;  C_k\, \sum_{\sigma\in S_k} 
  (-1)^\sigma\, (\Tr_{\C^q}\otimes\calT) \bigg(p \prod_{i=1}^k \partial_{\sigma(i)}p \bigg)\;, 
\end{equation*}
where $C_k=\frac{(2\pi i)^{k/2}}{(k/2)!}$ and $S_k$ is the permuation group of $\{1,\ldots,d\}$.
\end{theorem}

Like the setting with $k$ odd, the computation can be substantially simplified with some preliminary results. 
Let us again focus on the case $q=1$ and first recall the even local index formula~\cite{CPRS3}:
$$
  \langle [p],[X_\mathrm{even}]\rangle \;=\; \res_{r=(1-k)/2} \;\sum_{m=0,\text{even}}^k 
  \; \phi_m^r(\mathrm{Ch}^m(p))\; ,
$$
where $\phi_m^r$ is the resolvent cocycle and
$$    
\mathrm{Ch}^{2n}(p) \;=\; (-1)^n \frac{ (2n)!}{2(n!)}\,(2p-1)\otimes p^{\otimes 2n}\;,  
    \qquad \mathrm{Ch}^0(p)\; =\; p\;.
$$

\vspace{.1cm}

\begin{proof}[Proof of Theorem \ref{thm:higher_dim_chern_number_even}]
The proof of 
Lemma \ref{prop:only_need_top_term_in_local_index_odd_case} also holds 
here to show that $\phi_m^r(\mathrm{Ch}^m({p}))$ does not contribute 
to the index pairing for $0<m<k$.
Therefore the index computation is reduced to
\begin{align*}
  \left\langle [p],\left[X_\text{even}\right]\right\rangle &\;=\; \res_{r=(1-k)/2} \;\phi_k^r(\mathrm{Ch}^k(p))\;,
\end{align*} 
which is a residue at $r=(1-k)/2$ of the term
$$
 \mathcal{C}_k\int_0^\infty\!\! s^k\, \Tr_\tau \Big(\Gamma_0  \int_\ell \lambda^{-k/2-r} (2p-1) R_s(\lambda) [D,p] R_s(\lambda)\cdots [D,p]R_s(\lambda) \,\mathrm{d}\lambda\Big)\mathrm{d}s
  \;,
$$
where $\Gamma_0 = (-i)^{k/2}\Gamma^1\Gamma^2\cdots\Gamma^k$ is the grading 
operator of ${\C^\nu}$ and
$$
  \mathcal{C}_k \;=\; \frac{(-1)^{k/2}k! \, 2^{k}\Gamma(k/2+1)}{i\pi(k/2)! \,\Gamma(k+1)}\;
$$
comes from the resolvent cocycle and the normalisation of $\mathrm{Ch}^{k}(p)$.
Like the case of $k$ odd, one can move the resolvent terms to the right up to a holomorphic 
error in order to take the Cauchy integral. Lemma \ref{lem:residue_trace_formula_part1} 
implies that the complex function 
$\Tr_\tau \!\left(\Gamma_0(2{p}-1)([D,{p}])^k(1+D^2)^{-z/2}\right)$ has at 
worst a simple pole at $\Re(z)=k$. Computing the residue explicitly,
$$  
\res_{r=(1-k)/2} \phi_k^r(\mathrm{Ch}^k(p)) = \frac{(-1)^{k/2}}{2((k/2)!)}\sigma_{k/2,1}\,\res_{z=k} 
\Tr_\tau\!\left(\Gamma_0 (2{p}-1)([D,{p}])^k(1+D^2)^{-z/2}\right), 
$$
where $\sigma_{k/2, 1}$ is the coefficient of $z$ in $\prod_{j=0}^{k/2-1}(z+j)$ 
and is given by the number $\sigma_{k/2, 1} = ((k/2)-1)!$. Putting these results back together,
$$ 
\langle [p], [X_\mathrm{even}]\rangle =  
(-1)^{k/2}  \frac{1}{k} \res_{z=k} \Tr_\tau\!\left(\Gamma_0 (2{p}-1)([D,{p}])^k(1+D^2)^{-z/2}\right). 
$$

\vspace{.1cm}

Next we claim that $\Tr_\tau\!\left(\Gamma_0([D,{p}])^k(1+D^2)^{-z/2}\right) = 0$ for $\Re(z)>k$. 
To see this, let us compute for $\Gamma_0= (-i)^{k/2}\Gamma^1\cdots \Gamma^k$,
\begin{align*}
   [D,{p}]^k &\;=\;  \sum_{\sigma\in S_k}(-1)^\sigma \prod_{i=1}^k \,[X_{\sigma(i)}, p] \hat\otimes \Gamma^i 
     \;= \;i^{k/2}\Gamma_0 \sum_{\sigma\in S_k}(-1)^\sigma \prod_{j=1}^k \,[X_{\sigma(j)}, p] \hat\otimes 1_{{\C^\nu}}\;.
\end{align*}
Because $\sum_{\sigma}(-1)^\sigma \prod_{j=1}^k [X_{\sigma(j)}, p]$ is 
symmetric with respect to the $\pm 1$ eigenspaces of $\Gamma_0$, the 
spinor trace $\Tr_\tau(\Gamma_0[D,{p}]^k(1+D^2)^{-z/2})$ will vanish for 
$\Re(z)>k$. Therefore the zeta function $\Tr_\tau(\Gamma_0[D,{p}]^k(1+D^2)^{-z/2})$ 
analytically continues as a function holomorphic in a neighbourhood of 
$z=k$ and its residue does not contribute to the index.

\vspace{.1cm}

We know that 
$[D,{p}] = \sum_{j=1}^k [X_j,{p}]\hat\otimes\Gamma^j = i\sum_{j=1}^k \partial_j{p}\hat\otimes\Gamma^j$ and so
$$ 
  {p}([D,{p}])^k \;=\; (-1)^{k/2} {p} \sum_{\sigma\in S_k}(-1)^\sigma 
      \prod_{j=1}^k \partial_{\sigma(j)}{p}\hat\otimes \Gamma^{j}\;. 
$$
Therefore, recalling the spinor degrees of freedom 
and using Equation \eqref{eq:residue_trace_formula},
\begin{align*}
   &\langle [p], [X_\mathrm{even}]\rangle \;=\;   (-1)^{k/2}  \frac{1}{k} \res_{z=k} 
     \Tr_\tau\!\left(\Gamma_0 \,2{p}([D,{p}])^k(1+D^2)^{-z/2}\right) \\
    &\hspace{0.2cm}\;= 
    \;(-1)^{k/2}(-1)^{k/2}\frac{ i^{k/2}2^{k/2}}{k} \res_{z=k} 
    \Tr_{\tau} \bigg( {p} \sum_{\sigma\in S_k}(-1)^\sigma \prod_{j=1}^k \partial_{\sigma(j)}{p}(1+|X|^2)^{-z/2} \bigg) \\
    &\hspace{0.2cm}\;=\; \frac{(2i)^{k/2} \mathrm{Vol}_{k-1}(S^{k-1}) }{k}\, 
    \calT \bigg(p \sum_{\sigma\in S_k}(-1)^\sigma \prod_{j=1}^k \partial_{\sigma(j)}{p} \bigg).
\end{align*}
Lastly, we use that 
$\mathrm{Vol}_{k-1}(S^{k-1}) = \frac{k\pi^{k/2}}{(k/2)!}$ for $k$ even to simplify
\begin{equation*} 
   \langle [p], [X_\mathrm{even}]\rangle \;=\;   \frac{(2\pi i)^{k/2}}{(k/2)!} \sum_{\sigma\in S_k}(-1)^\sigma \, 
   \calT \bigg({p} \prod_{i=1}^k \partial_{\sigma(i)}{p} \bigg)\;,
\end{equation*}
and this concludes the proof.
 \end{proof}

\vspace{.1cm}

The even and odd index formulas recover the generalised Connes--Chern characters 
for crossed products studied in~\cite[Section 6]{PSBKK}. We emphasise that 
while we can construct both complex and real Kasparov modules and semifinite spectral triples, 
the local index formula only applies to complex algebras and invariants.

\subsection{Application to topological phases}

Here we return to the case of $A = \big(C(\Omega) \rtimes_\phi \Z^{d-k}\big)\rtimes_\theta\Z^k$
with $B=C(\Omega)\rtimes_\phi \Z^{d-k}$. If the algebra 
is complex and the system has no chiral symmetry, then the $K$-theory class of interest 
is the Fermi projection $P_F = \chi_{(-\infty,\mu]}(H)$, which is in $A$ 
under the gap assumption. 
If there is a chiral symmetry present, 
then $H$ can be expressed as $\begin{pmatrix} 0 & Q^* \\ Q & 0 \end{pmatrix}$ 
with $Q$ invertible (assuming the Fermi energy at $0$). Therefore one can take 
the so-called Fermi unitary $U_F = Q|Q|^{-1}$ and obtain a class in $K_1(A)$. Of 
course, this unitary is relative to the diagonal chiral symmetry operator 
$R_{ch} = \begin{pmatrix} 1 & 0 \\ 0 & -1 \end{pmatrix}$ and so the 
invariants are with reference to this choice, see~\cite{DNG15,Thiang14b} 
for more information on this issue.
Provided $H$ is a matrix of elements 
in $\calA$ (which is physically reasonable), then the above local formulas for the 
weak invariants will be valid.
 
\vspace{.1cm}

Firstly, if $k=d$ then the index formulae are the Chern numbers for the 
strong invariants studied in~\cite{PSBbook}. If the measure $\bP$ on $\Omega$ is 
ergodic under the $\Z^d$-action, then 
$\calT(a) = \Tr_{\mathrm{Vol}}(\pi_\omega(a))$ for almost all $\omega$, where  
$\Tr_{\mathrm{Vol}}$ is the trace per unit volume on $\ell^2(\Z^d)$ and 
$\{\pi_\omega\}_{\omega\in\Omega}$ is a family representations 
$C(\Omega)\rtimes_\phi \Z^d\to \calB(\ell^2(\Z^d))$ linked by a covariance relation~\cite{PSBbook}. 
Under the ergodicity hypothesis, the tracial formulae become
\begin{align*}
   \langle [U_F], [X_\mathrm{odd}] \rangle &\,=\, C_k \sum_{\sigma\in S_k}(-1)^\sigma\, 
 (\Tr_{\C^q}\otimes\Tr_{\mathrm{Vol}})\bigg(\prod_{i=1}^k \pi_\omega(U_F)^* (-i)[X_{\sigma(i)},\pi_\omega(U_F)]\bigg) 
\,,
\\
  \langle [P_F], [X_\mathrm{even}]\rangle &\,=\,  C_k \,\sum_{\sigma\in S_k}\! 
  (-1)^\sigma (\Tr_{\C^q}\otimes\Tr_{\mathrm{Vol}}) \bigg(\! \pi_\omega(P_F) 
     \prod_{i=1}^k (-i)[X_{\sigma(i)},\pi_\omega(P_F)]\bigg)\,,
\end{align*}
for almost all $\omega\in\Omega$. As the left hand side of the equations are independent of 
the disorder parameter $\omega$, the weak invariants are 
stable almost surely under the disorder. Recall that we require the Hamiltonian $H_\omega$
to have a spectral gap for all $\omega\in\Omega$, so our results do not apply to the 
regime of strong disorder where the Fermi projection lies in a mobility gap.

\vspace{.1cm}

The physical interpretation of our semifinite pairings has been discussed in~\cite{PSBbook}. 
For $k$ even, the pairing $\langle [P_F], [X_\mathrm{even}] \rangle$ can be linked to the 
linear and non-linear transport coefficients of the conductivity tensor of the 
physical system. For $k$ odd, the pairing $\langle [U_F], [X_\mathrm{odd}] \rangle$ 
is related to the chiral electrical polarisation and its derivates (with respect 
to the magnetic field). See~\cite{PSBbook} for more details. All algebras are separable, which 
implies that the semifinite 
pairing takes values in a discrete subset of $\R$. Hence we have proved that the 
physical quantities related to the semifinite pairings are quantised and topologically stable.

\vspace{.2cm}

\section{Real pairings and torsion invariants} \label{sec:real_pairings}
The local index formula is currently only valid for complex algebras and 
spaces. Furthermore, the semifinite index pairing involves taking a trace 
and thus it will vanish on torsion representatives, which are more common in the 
real setting. 
Because of the anti-linear symmetries that are of interest in topological 
insulator systems, we would also like a recipe to compute the pairings 
of interest in the case of real spaces and algebras.

\vspace{.1cm}

Given a disordered Hamiltonian $H\in M_n(C(\Omega)\rtimes \Z^d)$ (considered now
as a real subalgebra of a complex algebra) satisfying time-reversal or particle-hole 
symmetry (or both) and thus determining the symmetry class index $n$, 
one can associate a class $[H]\in KO_n(C(\Omega)\rtimes \Z^d)$ 
(see~\cite{Thiang14, Kellendonk15, Kubota15b,BCR15}). The class can then be paired with the 
unbounded Kasparov module $\lambda_k$ from Proposition \ref{prop:crossed_prod_kas_mod}. 
As outlined in Section \ref{sec:spin_and_oriented_modules}, we prefer to work with the 
Kasparov module $\lambda_k$ coming from the oriented structure 
$\ell^2(\Z^k,B)\hat\otimes \bigwedge^*\R^k$ 
as the Clifford actions are explicit and easier to work with.
In the case of a unital algebra $B$ and $A=B\rtimes \Z^k$, there is 
a well-defined map
\begin{align*}
    &KO_n(B\rtimes \Z^k) \times KKO^k(B\rtimes \Z^k, B)\; \to\; KKO(C\ell_{n,k},B)
\;.
\end{align*}
The class in $KKO(C\ell_{n,k}, B)$ can be represented by 
a Kasparov module  $(C\ell_{n,k}, E_B, \hat{X})$ which can be bounded or unbounded. 
Up to a finite-dimensional adjustment 
(see~\cite[Appendix B]{BCR15}), the topological information of interest of this 
Kasparov module is contained 
in the kernel, $\Ker(\hat{X})$, which is a finitely generated and projective $C^*$-submodule of $E_B$
with a graded left-action of $C\ell_{n,k}$. If $B$ is ungraded, an Atiyah--Bott--Shapiro 
like map then gives an isomorphism $KKO(C\ell_{n,k},B) \to KO_{n-k}(B)$ via Clifford modules, 
see~\cite[Section 2.2]{Schroeder}.

\vspace{.1cm}

Considering the example of $B=C(\Omega)\rtimes \Z^{d-k}$, then one has the 
Clifford module valued index
\begin{align*}
    &KO_n(C(\Omega)\rtimes \Z^d) \times 
    KKO^k\big( C(\Omega)\rtimes \Z^d, C(\Omega)\rtimes \Z^{d-k}\big)
    \;\to\; KO_{n-k}(C(\Omega)\rtimes \Z^{d-k})\;.
\end{align*}
If $k=d$, then the pairing takes values in $KO_{n-d}(C(\Omega))$ and 
constitute `strong 
invariants'. Furthermore, fixing a disorder configuration 
$\omega\in \Omega$ provides  a map $KO_{n-d}(C(\Omega)) \to KO_{n-d}(\R)$ 
and then a corresponding analytic index formula can be obtained as in~\cite{GSB15} 
(note, however,
that \cite{GSB15} also covers the case of a mobility gap  which does not
require a spectral gap).

\vspace{.1cm}

To compute range of the weak $K$-theoretic pairing, let us first 
consider the case of $\Omega$ contractible. Then one can compute directly
$$
  KO_{n-k}(C(\Omega)\rtimes \Z^{d-k}) \;\cong\; KO_{n-k}(C^*(\Z^{d-k})) \;\cong\; 
  \bigoplus_{j=0}^{d-k}\binom{d-k}{j} KO_{n-k-j}(\R)\;,
$$
which for the varying values of $k\in\{1,\ldots,d-1\}$ recovers 
the weak phases described for systems 
without disorder in Equation \eqref{eq:periodic_weak_phase}. 
Computing the range of the pairing for non-contractible $\Omega$ is much 
harder, see~\cite[Section 6]{Kellendonk16} for the computation of 
$KO_j(C(\Omega)\rtimes\Z^2)$ for low $j$.
Note also 
that a different action $\alpha'$ on $\Omega$ or a different 
disorder configuration space $\Omega'$ could potentially lead 
to different invariants.

\vspace{.1cm}

If the $K$-theory class $[x]\in KO_{0}(B)$ is not torsion-valued and 
$B$ contains a trace, then one may take the induced trace $[\tau_B(x)]$ and 
obtain a real-valued invariant. For $B=C(\Omega)\rtimes \Z^{d-k}$, 
the induced trace plays the role of averaging over the disorder and $(d-k)$
spatial directions. For non-torsion elements in $KO_{j}(B)$ 
with $j\neq 0$, we can apply the induced trace by rewriting
$KO_{j}(B) \cong KO_0(C_0(\R^j)\otimes B) \cong KKO(\R,B\hat\otimes C\ell_{0,j})$. 
This equivalence comes with the limitation that one either has to work 
with traces on suspensions or graded traces on Clifford algebras.
Of course, if $[x]$ is a torsion element the discussion does not apply
as $[\tau(x)]=0$. See~\cite{Kellendonk16} for recent 
work that aims to circumvent some of these problems.


\section{The bulk-boundary correspondence}
We consider the (real or complex) algebra $B\rtimes_\theta\Z^k$ with $k\geq 2$ and the twist 
$\theta$ such that $\theta(m,-m)=1$ for all $m\in\Z^k$~\cite{KRS,PSBbook}. Then one can decompose
$B\rtimes_\theta\Z^k \cong (B\rtimes_\theta \Z^{k-1})\rtimes\Z$, 
which gives us a short exact sequence of $C^*$-algebras
\begin{equation} \label{eq:Toeplitz_crossedprod_ext}
  0 \;\to\; (B\rtimes_\theta\Z^{k-1})\otimes \calK(\ell^2(\N)) 
\;\to \;
\calT_\Z \to B\rtimes_\theta\Z^k 
\;\to\; 0\;.
\end{equation}
The Toeplitz algebra $\calT_\Z$ for the crossed product is described in~\cite{KRS,BKR1,PSBbook}. 
In particular, the algebra $\calT_\Z$ acts on the $C^*$-module $\ell^2(\Z^{k-1}\times\N,B)$, thought 
of as a space with boundary and the ideal $(B\rtimes_\theta\Z^{k-1})\otimes \calK(\ell^2(\N))$ can 
be thought of as observables concentrated at the boundary $\ell^2(\Z^{k-1}\times\{0\},B)$. 

\vspace{.1cm}

Let $A_e = B\rtimes_\theta\Z^{k-1}$ be the edge algebra with 
bulk algebra $B\rtimes_\theta\Z^k = A_e \rtimes\Z$. Associated to Equation 
\eqref{eq:Toeplitz_crossedprod_ext} is a class in 
$\mathrm{Ext}^{-1}(A_e\rtimes\Z,A_e) \cong KKO^1(A_e\rtimes\Z,A_e)$ by~\cite[{\S}7]{Kasparov80}. 

\begin{proposition}[\cite{BKR1}, Proposition 3.3]
The Kasparov module $\lambda_1$ from Proposition 
\ref{prop:crossed_prod_kas_mod} with $k=1$ and representing 
$[\lambda_1]  \in KKO^1(A_e\rtimes\Z,A_e)$ or $KK^1(A_e\rtimes\Z,A_e)$ 
also represents the extension class of 
Equation \eqref{eq:Toeplitz_crossedprod_ext}.
\end{proposition}

Similarly, one can use Proposition~\ref{prop:crossed_prod_kas_mod} to build an edge
Kasparov module $\lambda_{k-1}$ representing a class in $KKO^{k-1}(B\rtimes_\theta\Z^{k-1},B)$ or $KK^{k-1}(B\rtimes_\theta,B)$.
Hence we have a map 
$$
KKO^1(B\rtimes \Z^k, B\rtimes \Z^{k-1}) \times KKO^{k-1}(B\rtimes \Z^{k-1},B) 
 \; \to\; KKO^k (B\rtimes \Z^k, B)
$$
given by the Kasparov product $[\lambda_1]\hat\otimes_{A_e} [\lambda_{k-1}]$ at the level of classes.

\begin{theorem}[\cite{BKR1}, Theorem 3.4] \label{thm:bulk-edge}
The product $[\lambda_1]\hat\otimes_{A_e} [\lambda_{k-1}]$ has the unbounded 
representative
$$
  \bigg( \calA \hat\otimes C\ell_{0,k},\, \ell^2(\Z^{k},B)_B \hat\otimes 
  \bigwedge\nolimits^{\!\ast}\R^{k},\,\, X_k \hat\otimes\gamma^1 + 
    \sum_{j=1}^{k-1} X_j \hat\otimes \gamma^{j+1} \bigg)
$$
and at the bounded level 
$[\lambda_1]\hat\otimes_{A_e} [\lambda_{k-1}]=(-1)^{k-1}[\lambda_k]$, where
$-[x]$ represents the inverse of $[x]$ in the $KK$-group.
\end{theorem}

Recall that the weak invariants arise from the pairing of $\lambda_k$ 
with a class $[H]\in KO_n(B\rtimes \Z^k)$ (or complex). 
Theorem \ref{thm:bulk-edge} 
implies that
\begin{align*}
  [H]\hat\otimes_{A} [\lambda_k] &\;= \;[H]\hat\otimes_A 
    \big( [\lambda_1]\hat\otimes_{A_e} [\lambda_{k-1}] \big)
  \; =\; (-1)^{k-1} \big( [H]\hat\otimes_{A} [\lambda_1] \big) \hat\otimes_{A_e} [\lambda_{k-1}]
\;,
\end{align*}
by the associativity of the Kasparov product. On the other hand, let us note that 
$[H]\hat\otimes_{A} [\lambda_1] = \partial[H]\in KO_{n-1}(A_e)$ 
as the product with $[\lambda_1]$ represents the boundary map in 
$KO$-theory associated to the short exact sequence of 
Equation \eqref{eq:Toeplitz_crossedprod_ext}. Hence the weak pairing, up 
to a possible sign, is the same as a pairing over the 
edge algebra $A_e= B\rtimes_\theta \Z^{k-1}$.

\begin{corollary}[Bulk-boundary correspondence of weak pairings]
The weak pairing $[H]\hat\otimes_{A} [\lambda_k]$ is non-trivial 
if and only if the edge pairing $\partial[H] \hat\otimes_{A_e} [\lambda_{k-1}]$ 
is non-trivial.
\end{corollary}

In the real case we achieve a bulk-boundary correspondence of the $K$-theoretic 
pairings representing the weak invariants. 
The Morita equivalence between spin and oriented structures means that 
Theorem \ref{thm:bulk-edge} also applies to the spin Kasparov module 
$\lambda_k^{\!\mathfrak{S}}$. In particular, the bulk-boundary correspondence extends 
to the semifinite pairing, allowing us to recover the following result from \cite{PSBbook}.

\begin{corollary}[Bulk-boundary correspondence of weak Chern numbers]
The cyclic expressions for the complex semifinite index pairing 
are the same (up to sign) for the bulk and edge algebras. Namely 
for $k\geq 2$ and $p,u\in M_q(\calA)$,
\begin{align*}
  &\langle [u], [X_\mathrm{odd}] \rangle \;=\; \langle \partial[u], [X_{\mathrm{even}}] \rangle\;,
  &&\langle [p],[X_\mathrm{even}] \rangle \;=\; - \langle \partial [p], [X_\mathrm{odd}] \rangle\;.
\end{align*}
\end{corollary}
\begin{proof}
Because the factorisation of pairings occurs at the level of the Kasparov 
modules $\lambda_k^{\!\mathfrak{S}}$, the result immediately follows when taking the 
trace.
\end{proof}

\vspace{.1cm}

Recall that for $B=C(\Omega)\rtimes_\phi \Z^k$, the complex $K$-theory classes 
of interest were the Fermi projection $P_F$ or the Fermi unitary coming from 
$\mathrm{sgn}(H) = \begin{pmatrix} 0 & \;U_F^* \\ U_F & \;0 \end{pmatrix}$ 
if $H$ is chiral symmetric. 
We take the edge algebra,
$A_e = \big(C(\Omega)\rtimes_\phi \Z^{d-k}\big)\rtimes_\theta \Z^{k-1} 
\cong C(\Omega)\rtimes_\phi \Z^{d-1}$, which is an algebra associated to a 
system of $1$ dimension lower. 
The boundary maps in $K$-theory $\partial[P_F]$ and $\partial[U_F]$ can be 
written in terms of the Hamiltonian $\wh{H}\in \calT_\Z$ associated to the 
system with boundary. Furthermore, the pairings 
$\langle \partial[P_F], [X_\mathrm{odd}] \rangle$ 
and $\langle \partial[U_F], [X_{\mathrm{even}}] \rangle$ can be related to 
edge behaviour of the sample with boundary, e.g. edge conductance, see~\cite{KRS,PSBbook}. Hence 
in the better-understood complex setting, the bulk-boundary correspondence has 
both physical and mathematical meaning.

\end{document}